\documentclass[11pt]{article}
\usepackage[utf8]{inputenc}
\usepackage[margin=1.0in]{geometry}
\usepackage{amsmath}
\usepackage{amssymb}
\usepackage{amsthm}
\usepackage{subfigure}
\usepackage[colorlinks=true,citecolor=orange]{hyperref}
\usepackage[linesnumbered,algoruled,boxed,lined]{algorithm2e}
\usepackage{algpseudocode}
\usepackage{graphics}
\usepackage{graphicx}
\usepackage{xcolor}
\usepackage{bm}
\usepackage{enumerate}
\usepackage{mathrsfs}
\usepackage{bbm}
\usepackage{float}
\usepackage{tikz}

\usepackage{thm-restate}
\usepackage[noabbrev,capitalize,nameinlink]{cleveref}
\usepackage[T1]{fontenc}
\usepackage{csquotes}

\newtheorem{theorem}{Theorem}
\newtheorem{lemma}{Lemma}

\theoremstyle{definition}
\newtheorem{definition}{Definition}

\newtheorem{assumption}{Assumption}

\newcommand{\eps}{\varepsilon}

\newcommand{\mb}{\mathbf}
\newcommand{\mc}{\mathcal}
\newcommand{\st}{\mathrm{s.t.}}

\newcommand{\diag}{\mathrm{diag}}
\newcommand{\poly}{\mathrm{poly}}

\newcommand{\wt}{\widetilde}

\newcommand{\opt}{\mathrm{OPT}}
\newcommand{\sgn}{\mathrm{sgn}}
\newcommand{\R}{\mathbb{R}}

\renewcommand{\O}{\wt{O}}
\renewcommand{\l}{\langle}
\renewcommand{\r}{\rangle}

\newcommand{\nnz}{\mathsf{nnz}}
\newcommand{\bDelta}{\boldsymbol{\Delta}}
\newcommand{\g}{\nabla}

\newcommand{\norm}[1]{\left\lVert#1\right\rVert}

\foreach \x in {a,...,z}{%
	\expandafter\xdef\csname B\x\endcsname{\noexpand\mathbf{\x}}
}

\foreach \x in {A,...,Z}{%
	\expandafter\xdef\csname B\x\endcsname{\noexpand\mathbf{\x}}
}

\foreach \x in {A,...,Z}{%
	\expandafter\xdef\csname C\x\endcsname{\noexpand\mathcal{\x}}
}

\newcommand{\defeq}{\stackrel{\mathrm{\scriptscriptstyle def}}{=}}

\makeatletter
\renewcommand*\env@matrix[1][c]{\hskip -\arraycolsep
  \let\@ifnextchar\new@ifnextchar
  \array{*\c@MaxMatrixCols #1}}
\makeatother

\title{High-Accuracy Multicommodity Flows via Iterative Refinement}
\date{\today}
\author{
Li Chen\thanks{Li Chen was supported by NSF Grant CCF-2106444.}\\ Georgia Tech\\ \texttt{lichen@gatech.edu}
\and
Mingquan Ye\thanks{Mingquan Ye was supported by NSF Grant CCF-2240024.}\\ University of Illinois at Chicago\\ \texttt{mye9@uic.edu}
}
\begin{document}
\maketitle
\thispagestyle{empty}
\begin{abstract}

The multicommodity flow problem is a classic problem in network flow and combinatorial optimization, with applications in transportation, communication, logistics, and supply chain management, etc. Existing algorithms often focus on low-accuracy approximate solutions, while high-accuracy algorithms typically rely on general linear program solvers. In this paper, we present efficient high-accuracy algorithms for a broad family of multicommodity flow problems on undirected graphs, demonstrating improved running times compared to general linear program solvers. Our main result shows that we can solve the $\ell_{q, p}$-norm multicommodity flow problem to a $(1 + \varepsilon)$ approximation in time $O_{q, p}(m^{1+o(1)} k^2 \log(1 / \varepsilon))$, where $k$ is the number of commodities, and $O_{q, p}(\cdot)$ hides constants depending only on $q$ or $p$. As $q$ and $p$ approach to $1$ and infinity respectively, $\ell_{q, p}$-norm flow tends to maximum concurrent flow.

We introduce the first iterative refinement framework for $\ell_{q, p}$-norm minimization problems, which reduces the problem to solving a series of decomposable residual problems. In the case of $k$-commodity flow, each residual problem can be decomposed into $k$ single commodity convex flow problems, each of which can be solved in almost-linear time.
As many classical variants of multicommodity flows were shown to be complete for linear programs in the high-accuracy regime [Ding-Kyng-Zhang, ICALP'22], our result provides new directions for studying more efficient high-accuracy multicommodity flow algorithms.

\end{abstract}
\newpage

\section{Introduction}

The multicommodity flow problem is a classic challenge in network flow and combinatorial optimization, where the goal is to optimally route multiple commodities through a network from their respective sources to their respective sinks, subject to flow conservation constraints. This problem has significant applications in various fields such as transportation, communication, logistics, and supply chain management~\cite{kennington1978survey,ahuja1988network,ouorou2000survey,barnhart2009multicommodity,wang2018multicommodity}. Currently, the fastest algorithms for computing high-accuracy solutions involve formulating these problems as linear programs and employing generic linear program solvers~\cite{khachiyan1980polynomial,karmarkar1984new,renegar1988polynomial,cls21}. Notably, linear programs can be reduced to multicommodity flow problems with near-linear overhead~\cite{itai1978two, dkz22}.

Existing research predominantly focuses on obtaining $(1+\varepsilon)$-approximate solutions for maximum concurrent $k$-commodity flows~\cite{lsmtpt91,Young95,Fleischer00,Kar02,GargK07,Nesterov09,Madry10}, as summarized in \cref{tab:kcommLi}. However, these low-accuracy algorithms feature running times that are polynomial in $1/\varepsilon$ for computing $(1+\varepsilon)$-approximate solutions. In contrast, high-accuracy algorithms exhibit running times that are polynomial in $\log(1/\varepsilon)$. Importantly, \cite{dkz22} demonstrated that any enhancement in the high-accuracy algorithm for the 2-commodity flow problem would result in a faster general linear program solver.

In this paper, we investigate multicommodity flow problems on undirected graphs, which possess more structure than their directed counterparts. Prior work has shown that maximum concurrent 2-commodity flow on undirected graphs can be reduced to two instances of maximum flow problems, both solvable in almost-linear time~\cite{rw66,cklpgs22}. More generally, maximum concurrent $k$-commodity flows can be reduced to $2^{k-1}$ maximum flows\footnote{To see this, the edge capacity constraint $\sum_j |f_{ej}| \le \Bu_e$ is equivalent to $|\sum_j \eps_j f_{ej}| \le \Bu_e$ for any $\eps \in \{\pm 1\}^k.$ }. Additionally, researchers have discovered that $(1+\varepsilon)$-approximate algorithms for undirected graphs are considerably faster than those for directed graphs~\cite{kmp12,klos14,sherman17}. Nevertheless, the fastest high-accuracy algorithms still rely on general linear program solvers. Given these advancements, we pose the following natural question:

\begin{displayquote}
\emph{Is it possible to solve multicommodity flow problems on undirected graphs to high accuracy more efficiently than with general linear program solvers?}
\end{displayquote}

This paper gives an affirmative answer and presents high-accuracy algorithms for a large family of multicommodity flow problems that runs in time $m^{1+o(1)}\poly(k, \log(1 / \varepsilon))$.
Our main result is an algorithm that, for any $1 \le q \le 2 \le p$ with $p = \O(1)$\footnote{We use $\O(f(n))$ to hide $\poly \log f(n)$ factors.} and $\frac{1}{q-1} = O(1)$, given edge weights $\Bw \in \R^E_+$ and $k$ vertex demands $\BD = [\Bd_1, \cdots, \Bd_k] \in \R^{V \times k}$, solves the following optimization problem to a $(1 + \varepsilon)$ approximation with high probability\footnote{We use ``with high probability'' (w.h.p.) throughout to say that an event happens with probability at least $1 - n^{-C}$ for any constant $C > 0.$} in time $O_{q, p}(m^{1+o(1)} k^2 \log(1 / \varepsilon))$:
\begin{align*}
    \min_{\text{$k$-commodity flow $\BF$ with residue $\BD$}} \norm{\BW \BF}_{q, p}^{pq} \defeq \sum_e \Bw_e^{pq} \left(\sum_j |\BF_{ej}|^q\right)^p
\end{align*}

The problem generalizes the maximum concurrent flow problem by setting the edge weights $\Bw_e$ as the reciprocal of edge capacities and letting $q \to 1, p \to \infty$. 
Therefore, $\ell_{q, p}$-norm flows are natural relaxations of the combinatorial maximum concurrent flows.
However, unlike the typical relaxations using the exponential function ($\exp(\text{congestion})$) in previous efficient approximation schemes, we show that $\ell_{q, p}$-norm problems themselves admit high accuracy solutions.
Thus, we provide a large family of multicommodity flows that admit high-accuracy solutions in almost linear time.

From a technical standpoint, our exploration of multicommodity flows reflects the research trajectory of $\ell_p$-norm single commodity flows in recent years~\cite{kpsw19,ABKS21}. This line of study has led to the development of several novel algorithmic components, some of which have proven beneficial for classical single-commodity flow as well~\cite{kls20, AMV20}. More significantly, examining $\ell_p$-norm flows, particularly their weighted variants, has directed attention towards the core challenges of flow problems. We posit that further investigation of $\ell_{q, p}$-norm flows is likely to yield similar insights, potentially for variants of multicommodity flows not known to be hard, such as unit-capacity maximum concurrent flows on undirected graphs.


To summarize, this paper introduces the first iterative refinement framework for solving $\ell_{q, p}$-norm minimization problems. The proposed framework reduces the problem to approximately resolving $O_{p, q}(k \log(1 / \varepsilon))$ instances of decomposable residual problems. For the $k$-commodity flow case, each residual problem can be divided into $k$ single commodity flow problems and solved in $km^{1+o(1)}$-time using the almost-linear time convex flow solver~\cite{cklpgs22}.
As many classical variants of multicommodity flows were shown to be complete for linear programs in the high-accuracy regime~\cite{dkz22}, our result provides new directions for studying more efficient high-accuracy multicommodity flow algorithms.



\subsection{Main Result}
Consider an undirected graph $G = (V, E)$ and parameters $q$ and $p$. The \emph{$\ell_{q, p}$-norm multicommodity flow} problem asks for a multicommodity flow $\BF \in \R^{E \times k}$ that routes the given demands while minimizing the following objective:
\begin{align}
\label{eq:LqpFlow}
\min_{\BB^\top \BF = \BD} \norm{\BW \BF}_{q, p}^{pq} \defeq \sum_e \Bw_e^{pq} \left(\sum_j |\BF{ej}|^q\right)^p
\end{align}
Here, we have:
\begin{enumerate}
\item $\BB \in \R^{E \times V}$, the edge-vertex incidence matrix of $G$.
\item $\BD= [\Bd_1, \cdots, \Bd_k] \in \R^{V \times k}$, representing the set of $k$ vertex demands, where $\l\Bd_j, \mb{1}_n\r = 0$ for each $j \in [k]$.
\item $\Bw \in \R^{E}+$, the vector of edge weights, and $\BW = \diag(\Bw)$.
\end{enumerate}

This problem can be seen as a generalization of the classical maximum concurrent flow problem, which aims to find a feasible flow $\BF$ that minimizes edge congestion:

\begin{align}
\label{eq:maxConFlow}
\min_{\BB^\top \BF = \BD} \max_{e} \frac{1}{\Bc_e} \sum_{j \in [k]} |\BF_{ej}| = \norm{\BC^{-1} \BF}_{1, \infty}
\end{align}

Here, $\Bc \in \R^{E}_+$ denotes the vector of edge capacities and $\mb{C} = \diag(\mb{c})$. Our generalization allows for fractional vertex demands, meaning that for each commodity $j$, the demand $\Bd_j$ is not restricted to source-sink pairs.

The main technical result of this paper is presented below:
\begin{restatable}{theorem}{mainLqpFlow}\label{thm:LqpFlow}
Given any $1 \le q \le 2 \le p$ with $p = \O(1)$ and $\frac{1}{q-1} = O(1)$ and an error parameter $\varepsilon > \exp(-\O(1))$, let $\opt$ be the optimal value to Problem~\eqref{eq:LqpFlow}.
There is a randomized algorithm that computes a feasible flow $\BF$ to Problem~\eqref{eq:LqpFlow} such that
\begin{align*}
    \norm{\BW \BF}_{q, p}^{pq} \le \left(1 + \varepsilon\right) \opt + \varepsilon
\end{align*}
The algorithm runs in time
\begin{align*}
    O\left(p^2 \cdot \left(\frac{1}{q - 1}\right)^{\frac{1}{q - 1}} m^{1 + o(1)} \cdot k^2 \cdot \log \frac{1}{\varepsilon}\right)
\end{align*}
with high probability.
\end{restatable}

\subsection{Related Works}

In this section, we discuss some work related to the problem and the techniques we use.

\begin{table}[htbp]
    \centering
    \begin{tabular}{|l|c|l|c|c|}\hline 
        Year & References & Time & Directed? \\ \hline
        1990 & \cite{shahrokhi1990maximum} & $O(nm^7/\varepsilon^5)$ & Directed \\ \hline
        1991 & \cite{lsmtpt91} & $\O(mnk/\varepsilon^3)$ & Directed \\ \hline
        1996 & \cite{grigoriadis1996approximate} & $\O(mnk/\varepsilon^2)$ & Directed \\ \hline
        1996 & \cite{radzik1996fast} & $\O(mnk/\varepsilon^2)$ & Directed \\ \hline
        2009 & \cite{Nesterov09} & $\O(k^2m^2 / \varepsilon)$  & Directed \\ \hline
        2010 & \cite{Madry10} & $\O((m + k)n / \varepsilon^2)$ & Directed \\ \hline
        2012 & \cite{kmp12} & $\O(m^{4/3}\poly(k, \frac{1}{\varepsilon}))$ & Undirected \\ \hline
        2014 & \cite{klos14} & $m^{1+o(1)}k^2 / \varepsilon^2$ & Undirected \\ \hline
        2017 & \cite{sherman17} & $\O(mk/\varepsilon)$ & Undirected \\ \hline
        2019 & \cite{cls19STOC} & $(mk)^{\omega + o(1)} \log \frac{1}{\varepsilon}$ & Directed \\ \hline
    \end{tabular} 
    \caption{A summary of algorithms for the max concurrent $k$-commodity flow problem.}
    \label{tab:kcommLi}
\end{table}

\paragraph{Multicommodity Flow.}

The multicommodity flow problem is a classic problem in network flow and combinatorial optimization. Multicommodity flow has a wide range of applications in various fields which are addressed in numerous surveys~\cite{kennington1978survey,ahuja1988network,ouorou2000survey,barnhart2009multicommodity,wang2018multicommodity}.
These problems can be formulated as linear programs and solved using generic linear program solvers which remain the fastest algorithms for computing high-accuracy solutions~\cite{khachiyan1980polynomial,karmarkar1984new,renegar1988polynomial,cls21}.
On the other hand, linear programs can be reduced to $2$-commodity flow efficiently~\cite{itai1978two, dkz22}.
Specifically, any linear program can be reduced to a maximum throughput $2$-commodity flow instance with a near-linear overhead.

Much of the existing works focus on finding $(1+\varepsilon)$-approximate solutions.
\cite{shahrokhi1990maximum} gave the first FPTAS to the maximum concurrent flow problem with unit capacity (Problem~\eqref{eq:maxConFlow}).
Subsequently, a series of work~\cite{lsmtpt91,goldberg1992natural,grigoriadis1994fast,klein1994faster,karger1995adding,plotkin1995fast,grigoriadis1996approximate,radzik1996fast,shmoys1997cut} based on Lagrangian relaxation and linear program decomposition gave algorithms with improved running times for various version of the problem with arbitrary edge capacity.
These algorithms iteratively update the current flow and make progress by computing a series of either shortest path~\cite{shahrokhi1990maximum,klein1994faster,plotkin1995fast,shmoys1997cut}, or single commodity minimum cost flows~\cite{lsmtpt91,goldberg1992natural,grigoriadis1994fast,karger1995adding,grigoriadis1996approximate,radzik1996fast}.
In particular, \cite{radzik1996fast} and \cite{grigoriadis1996approximate} showed that finding $(1 + \varepsilon)$-approximate solutions can be reduced to $\O(k/\varepsilon^2)$ min-cost flow computations which take $\O(kmn/\varepsilon^2)$-time in total using the fastest min-cost flow algorithm at the time.
Combining with the almost-linear time min-cost flow algorithm yields a randomized $m^{1+o(1)} k / \varepsilon^2$-time max concurrent flow algorithm.

Later, a series of works based on multiplicative weight updates (MWU) gave conceptually simpler and faster algorithms at their time~\cite{Young95,Fleischer00,Kar02,GargK07,Madry10}.
These methods build the solution from scratch without re-routing the current flow.
At each step, they augment the current flow via a shortest path computation that favors relatively uncongested paths.
\cite{Madry10} used dynamic APSP data structure to speed up these computations and resulted in a $\O((m+k)n / \varepsilon^2)$-time max concurrent flow algorithm.
At the time, the fastest max concurrent flow runs in $\O(m^{1.5} k / \varepsilon^{2})$-time due to the $\O(m^{1.5})$-time min cost flow algorithm by \cite{ds08}.
The result of \cite{Madry10} was a significant improvement in the case when $k$ is large.

Another line of work focused on improving the $1 / \varepsilon^2$ term.
\cite{bienstock2004solving} found a FPTAS that only has $O(\frac{1}{\varepsilon \log 1/\varepsilon})$ dependence.
Later, \cite{Nesterov09} gave a FPTAS with only $O(1 / \varepsilon)$ dependence.
In particular, the algorithm by Nesterov runs in $\O(k^2m^2 / \varepsilon)$-time.

Similar to the situation of single commodity flows, researchers have discovered approximate algorithms with $m^{1.5 - O(1)}$ dependence on undirected graphs.
\cite{kmp12} gave the first $\O(m^{4/3}\poly(k, 1/\varepsilon))$ algorithm for max concurrent flow.
The algorithm implements the method of \cite{lsmtpt91} using \emph{electrical capacity-constrained flows} instead of min-cost flows.
Each electrical capacity-constrained flow can be reduced to, using width-reduction MWU~\cite{ckmst11}, $\O(m^{1/3}\poly(k, 1/\varepsilon))$ graph Laplacian systems.
In the breakthrough result of \cite{klos14}, they improved the running time to $m^{1+o(1)}k^2 / \varepsilon^2$ based on non-Euclidean gradient descents and fast oblivious routing.
Specifically, the algorithm computes an oblivious routing of congestion $m^{o(1)}$ and uses it to reduce the number of gradient descent iterations to $m^{o(1)} \cdot k / \varepsilon^2.$
Later, \cite{sherman17} introduced the idea of \emph{area convexity} and improved the iteration count to $\O(1/\varepsilon).$
This resulted in the first $\O(mk/\varepsilon)$-time max concurrent flow algorithm.

\paragraph{$\ell_p$-Norm Regression.}

The $\ell_p$-norm regression problem seeks to find a vector $\Bx$ that minimizes $\|\BA \Bx - \Bb\|_p$, where $\BA \in \R^{d \times n}$ and $\Bb \in \R^d$.
Varying in $p$ interpolates between linear regression ($p = 2$) and linear program ($p \in \{1, \infty\}$).
$\ell_p$-norm regression has gained significant attention in the past decade due to its wide range of applications and its implications for other convex optimization problems~\cite{meng2013robust,woodruff2013subspace}.
$\ell_p$-norm regression has drawn much attention in the past decade due to its wide range of applications and its implication to other convex optimization problems~\cite{meng2013robust,woodruff2013subspace}.
Many works have focused on low-accuracy algorithms for overconstrained matrices, i.e., $n \gg d$~\cite{dasgupta2009sampling,meng2013robust,woodruff2013subspace,clarkson2016fast,clarkson2017low,clarkson2019dimensionality}.
These results show various ways to find another matrix $\wt{\BA}$ with fewer rows such that $\|\wt{\BA} \Bx\|_p \approx \|\BA \Bx\|_p$ for any $\Bx$.
Then, approximate $\ell_p$-norm regression can be reduced to $\wt{\BA}$ and solved in time $O(\nnz(\BA) + \poly(d, 1/\varepsilon))$ when $p$ is a constant.

High accuracy solutions can be found using interior point methods (IPM) in $\O(\sqrt{n})$ or $\O(\sqrt{d})$ iterations~\cite{lee2019solving}.
\cite{bubeck2018homotopy} showed an homotopy method that finds high accuarcy solution in $\O(n^{|1/2 - 1/p|})$ iterations of linear system solves.
Based on the idea of iterative refinement and width reduction, a series of works~\cite{akps19,APS19,AdilS20,AKPS22} obtained improved iteration complexities of $\O_p(n^{(p - 2) / (3p - 2)})$ for $p \ge 2.$
Motivated by the success of sparse linear system solver~\cite{peng2021solving} and linear program in matrix multiplication time~\cite{cls21}, \cite{GPV21} gave a high accuracy algorithm that runs in $\O_p(n^{\theta})$ for some $\theta < \omega - \Omega(1)$, where $\omega$ is the matrix multiplication time exponent.

\paragraph{$\ell_p$-Norm Flows}
The $\ell_{q, p}$-norm formulation can be viewed as a multicommodity extension of the $\ell_p$-norm flows, which seeks to find a flow $\Bf \in \R^E$ that routes the given demand and minimizes $\|\diag(\Bw) \Bf\|_p$ where $\Bw \in \R^{E}_+.$
Varying in $p$ interpolates between the transshipment ($p=1$), the electrical flow ($p=2$), and the maximum flow problem ($p=\infty$).
Combining the result on $\ell_p$-norm regressions and Laplacian solvers, \cite{akps19} gave a $\O_p(m^{1 + |p-2|/(2p + |p-2|)})$-time $\ell_p$-norm flow algorithm.
Opening up the black box of the Spielman-Teng Laplacian solver~\cite{SpielmanTengSolver:journal}, \cite{kpsw19} gave the first $\O_p(m^{1 + O(1/\sqrt{p})})$-time high-accuracy $\ell_p$-norm flow algorithm for unweighted graphs, i.e., $\Bw = \mb{1}.$
The runtime is almost linear for $p = \omega(1)$ and this was the first almost-linear time high-accuracy algorithm for a large family of single commodity flow problems.
In the weighted case, \cite{ABKS21} combined \cite{kpsw19} and the idea of sparsification and gave a $p(m^{1+o(1)} + n^{4/3+o(1)})$-time high accuracy $\ell_p$-norm flow algorithm.
The study of the $\ell_p$-norm flow algorithms has been proved useful for single commodity flow problems such as unit-capacity maximum flows, bipartite matchings, and min-cost flows~\cite{kls20, AMV20}.



\paragraph{Continuous Optimization on Graphs.}

From a technical point of view, our $\ell_{q, p}$-norm multicommodity flow algorithm is inspired by the recent trend of applying continuous optimization techniques to solve graph problems.
As one of the earlier results in this direction, \cite{ds08} combined interior point methods and fast graph Laplacian system solvers~\cite{SpielmanTengSolver:journal} and gave a $\O(m^{3/2})$-time min cost flow algorithm.
Later, the idea culminated in a decade of works improving max flow and min-cost flow algorithms~\cite{ckmst11,madry13,sherman13,klos14,Mad16,peng16,cohen2017negative,AMV20,BLNPSSW20,kls20,bllsssw21,axiotis2022faster,bernstein2022deterministic,cklpgs22,glp21,bgjllps22}.

Beyond classical flow problems, the idea of combining continuous and combinatorial techniques gives improved algorithms for approximate shortest paths in parallel/distributed setting~\cite{Li20,zuzic2022universally}, faster network flow algorithms in distributed setting~\cite{forster2022minor}, flow diffusion~\cite{CPW21}, $\ell_p$-norm flows~\cite{kpsw19,ABKS21}, and more.

\subsection{Our Approach}

For the clarity of the presentation, we focus on the unweighted version of Problem~\eqref{eq:LqpFlow}, i.e., $\Bw_e = 1$ for each edge $e$, which is as follows:
\begin{align*}
    \min_{\BB^\top \BF = \BD} \norm{\BF}_{q, p}^{pq} = \sum_{e \in G} \left(\sum_{j=1}^k |\BF_{ej}|^q\right)^p. 
\end{align*}
The algorithm follows an overall \emph{iterative refinement} framework.
That is, given the current flow $\BF$, we want to find an update direction $\bDelta$ so that $\|\BF + \bDelta\|_{q,p}^{pq}$ is smaller than $\|\BF\|_{q,p}^{pq}.$
However, finding the optimal $\bDelta$ is equivalent to the original problem.
The idea of iterative refinement is to find a proxy \emph{residual} function $\CR(\bDelta; \BF)$ that approximates the Bregman divergence of $\|\BF\|_{q,p}^{pq}$, i.e.
\begin{align*}
    \|\BF + \bDelta\|_{q,p}^{pq} - \|\BF\|_{q,p}^{pq} - \left\l\BG, \bDelta\right\r \approx  \CR(\bDelta; \Bf)
\end{align*}
where $\BG \defeq \frac{\mathrm{d}}{\mathrm{d} \BF} \|\BF\|_{q,p}^{pq} \in \R^{E \times k}$ is the gradient.
Then, we can compute the direction $\bDelta$ by solving the \emph{residual problem}:
\begin{align*}
    \min_{\BB^\top \bDelta = \mb{0}} \left\l\BG, \bDelta\right\r + \CR(\bDelta; \BG) \approx \|\BF + \bDelta\|_{q,p}^{pq} - \|\BF\|_{q,p}^{pq}
\end{align*}
If $\CR(\bDelta; \BF)$ is a good approximation to the Bregman divergence, i.e. has a ``condition number'' of $\kappa$, we would obtain an $(1+\varepsilon)$-approximate solution in $O(\kappa \log(1 / \varepsilon))$ iterations.
On the other hand, $\CR(\bDelta; \BF)$ should be computationally easier to minimize so that we can implement each iteration efficiently.

For any $p > 1$, \cite{akps19} shows that $|x+\delta|^p$ can be locally approximated by a linear term plus an error term $\gamma_p(\delta; |x|)$, which behaves quadratically in $\delta$ when $|\delta| \le |x|$ and as $|\delta|^p$ otherwise.
Our key lemma (\cref{lem:diff_upper_lower}) extends the observation to approximating $\|\BF + \bDelta\|_{q,p}^{pq}$ and gives
\begin{align*}
\|\BF + \bDelta\|_{q,p}^{pq} - \|\BF\|_{q,p}^{pq} - \left\l\BG, \bDelta\right\r
\approx O_{p, q}(k) \sum_{e, j} \norm{\Bf_{e}}_q^{q(p-1)} \gamma_q(\bDelta_{ej}; \BF_{ej}) + O_{p, q}(k^{p}) \sum_{e, j} |\bDelta_{ej}|^{pq}
\end{align*}
Intuitively, the Bregman divergence can be approximated by a decomposable function on each coordinate $(e, j).$
The contribution of each coordinate $(e, j)$ behaves differently depending on the absolute value of $\bDelta_{ej}.$
It behaves quadratically when $|\bDelta_{ej}| \le |\BF_{ej}|$, as $|\bDelta_{ej}|^q$ when $|\bDelta_{ej}| > |\BF_{ej}|$ but smaller than $\|\Bf_e\|_q^q$, the total flow value on the edge $e$, and as $|\bDelta_{ej}|^{pq}$ otherwise.
The factors $k$ and $k^p$ comes from that given any $k$-dimensional vector $\Bx \in \R^k$, we have
\begin{align*}
    \norm{\Bx}_p^p \le \norm{\Bx}_1^p \le k^{p-1}\norm{\Bx}_p^p
\end{align*}

Surprisingly, this approximation has a conditioner number of $O_{q, p}(k)$ and is decomposable for each commodity $j$.
To obtain an $(1+\varepsilon)$-approximate solution, we only need to solve $O_{q, p}(k \log(1 / \varepsilon))$ iterations of residual problems of the form
\begin{align*}
\min_{\BB^\top \bDelta = \mb{0}} \left\l\BG, \bDelta\right\r + O_{p, q}(k) \sum_{e, j}\norm{\Bf_{e}}_q^{q(p-1)} \gamma_q(\bDelta_{ej}; \BF_{ej}) + O_{p, q}(k^{p}) \sum_{e, j} |\bDelta_{ej}|^{pq}
\end{align*}
The decomposability allows us to use the convex flow solver from \cite{cklpgs22} and solve the residual problem to high accuracy for each commodity in almost-linear time.
Thus, each iteration takes $km^{1+o(1)}$-time and the final running time is $O_{q, p}(k^2 m^{1+o(1)} \log(1/\varepsilon)).$

\subsection{Future Directions}


There are several open questions and next steps arising from the work.
\paragraph{Improving Dependency on $k$.}
Our algorithm (\cref{thm:LqpFlow}) exhibits an iteration complexity of $O_{q, p}(k \log(1/\eps))$, with each iteration solving $k$ single commodity convex flow problems. This results in a $k^2$ dependency in the overall runtime. In contrast, previous work on approximate maximum concurrent flows either reduced the problem to $\O(k)$ single commodity flow problems~\cite{lsmtpt91,grigoriadis1996approximate,radzik1996fast} or required $\O(1)$ iterations~\cite{sherman17}, yielding a linear dependency on $k$. It would be desirable to develop an $\ell_{q,p}$-norm flow algorithm that reduces the problem to only $O_{q, p}(k \log(1/\eps))$ instances of single commodity convex flow problems. One potential approach could involve solving for a single commodity in each iteration before updating the residual problem, rather than solving for all $k$ commodities at every iteration. Achieving this would necessitate a new iterative refinement lemma for the $\ell_{q,p}$-norm problem.

\paragraph{Improving Dependency on $q$.}
Our algorithm's runtime dependency on $q \in (1,2)$ is exponential, specifically $(\frac{1}{q-1})^{\frac{1}{q-1}}$. A similar dependency appears in the work on $\ell_q$-norm regressions~\cite{akps19, AKPS22}. The primary challenge lies in employing the $\gamma_q(x; 1)$ function to approximate $|1+x|^q$. Some research on $\ell_q$-norm regressions has sought to address this issue by focusing on the dual problem, which involves an $\ell_p$-norm regression problem for certain $p > 2$~\cite{ABKS21,JLS22}. However, in our case, adopting the dual approach does not eliminate the $\ell_q$-norm component due to the composite of norms present in our objective \eqref{eq:LqpFlow}. Overcoming the exponential dependency would require a more in-depth understanding of the iterative refinement framework for $\ell_q$-norms.

\paragraph{Unit-capacity Maximum Concurrent Flows.}
Our result on $\ell_{q,p}$-norm flows mirrors the line of work on $\ell_p$-norm single commodity flows~\cite{kpsw19,ABKS21}. Investigating $\ell_p$-norm flows has led to the development of algorithmic building blocks that have later proven beneficial for high-accuracy single commodity flow problems, such as unit-capacity maximum flows and uncapacitated min-cost flows~\cite{kls20,AMV20}. A natural question, then, is whether the $\ell_{q,p}$-norm flow algorithm could be leveraged to design more efficient algorithms for variants of multicommodity flows that are not yet considered difficult, such as unit-capacity maximum concurrent flows on undirected or even directed graphs.

\subsection{Paper Organization}
In \cref{sec:preli}, we introduce some preliminaries before presenting the technical parts, including the convex flow solver~\cite{cklpgs22} that is the core tool for solving the residual problem and the iterative refinement framework shown in \cite{akps19,AKPS22}. We formally present the proposed $\ell_{q,p}$-norm $k$-commodity flow algorithm in  \cref{sec:iterative_refinement}, and then solve the residual problem in \cref{sec:residualSolve}. 

\section{Preliminaries}
\label{sec:preli}
\subsection{General Notations}

We denote vectors (resp. matrices) by boldface lowercase (reps. uppercase) letters. For two vectors $\Bx, \By$, the vector $\Bx \cdot \By$ represents the entrywise product, i.e., $(\Bx \cdot \By)_i = \Bx_i\By_i$. Besides, for a vector $\Bx$, let $|\Bx|$ and $\Bx^p$ denote the entrywise absolute value and entrywise power of $\Bx$ respectively, that is, $|\Bx|_i = |\Bx_i|$ and $(\Bx^p)_i = (\Bx_i)^p$. We use $\l \Bx, \By \r$ to denote the inner product of $\Bx, \By$, i.e., $\l \Bx, \By \r = \Bx^\top\By$. For a vector $\Bx$, let $\diag(\Bx)$ represent a diagonal matrix whose $i$-th entry is equal to $\Bx_i$. 

\paragraph{Graphs.} In this paper, we consider multi-graphs $G$, with edge set $E(G)$ and vertex set $V(G)$. When the graph is clear from context, we use the short-hands $E$ for $E(G)$, $V$ for $V(G)$, $m = |E|, n = |V|$. We assume that each edge $e \in E$ has an implicit direction, used to define its edge-vertex incidence matrix $\BB \in \R^{E \times V}$. Abusing notation slightly, we often write $e = (u,v) \in E$ where $e$ is an edge in $E$ and $u$ and $v$ are the tail and head of $e$ respectively (note that technically multi-graphs do not allow for edges to be specified by their endpoints). 

We say a flow $\Bf \in \R^E$ routes a demand $\Bd \in \R^V$ if $\BB^\top\Bf=\Bd$. 
Given a multicommodity flow $\Bf \in \R^{E \times k}$, we use $\Bf_j \in \R^E, j \in [k]$ to denote the $j$-th column of $\Bf$, the flow for the commodity $j$, and $\Bf_e \in \R^k, e \in E$ to denote the row of $\Bf$ corresponding to the edge $e$, a vector of flows on $e.$

\paragraph{Model of Computation.} In this paper, for problem instances encoded with $z$ bits, all algorithms work in fixed-point arithmetic where words have $O(\log^{O(1)} z)$ bits, i.e., we prove that all numbers stored are in $[\exp(-\log^{O(1)}z), \exp(\log^{O(1)}z)]$.


\subsection{Convex Flow Solver}

In this paper, we utilize the almost-linear time convex flow algorithm from \cite{cklpgs22}.
Given a set of \emph{computationally efficient} convex cost functions on edges $\{c_e(\cdot)\}_e$, the algorithm finds a single commodity flow $\Bf$ that routes the given demand $\Bd$ and minimizes $\sum_e c_e(\Bf_e)$ up to a small $\exp(-\log^{O(1)} m)$ additive error.

\begin{assumption}[Definition 10.1 and Assumption 10.2, \cite{cklpgs22}]
\label{ass:cvxCost}
Let $K = \O(1)$ be a parameter fixed throughout.
Given a convex cost function $c: \R \to \R \cup \{+\infty\}$, $c$ is \emph{computationally efficient} if there is a barrier function $\psi_c(f, y)$ defined on the domain $\CD_c \defeq \{(f, y) |~ c(f) \le y\}$ such that
\begin{enumerate}
\item\label{item:costQuasiPoly}
The cost is quasi-polynomially bounded, i.e., $|c(f)| = O(m^K + |f|^K)$ for all $f \in \R.$ 
\item\label{item:costSC}
$\psi_c$ is a \emph{$\nu$-self-concordant barrier} for some $\nu = O(1)$, that is, the following holds
\begin{align*}
    \psi_c(\Bx) &\to \infty, \text{ as $\Bx$ approaches the boundary of $\CD_c$} \\
    \left|\g^3 \psi_c(\Bx) [\Bv, \Bv, \Bv]\right| &\le 2 \left(\g^2 \psi_c(\Bx) [\Bv, \Bv]\right)^{3/2}, \forall \Bx \in \CD_c, \Bv \in \R^2 \\
    \l\g \psi_c(\Bx), \Bv\r^2 &\le \nu \cdot \g^2 \psi_c(\Bx) [\Bv, \Bv]
\end{align*}
\item\label{item:costHessian}
The Hessian is quasi-polynomially bounded as long as the function value is $\O(1)$ bounded, i.e., for all points $|f|, |y|\le m^K$ with $\psi_c(f, y)\le \O(1)$, we have $\g^2 \psi_c(f, y) \preceq \exp(\log^{O(1)} m)\BI$. 
\item\label{item:costHessianCompute}
Both $\g \psi_c$ and $\g^2 \psi_c$ can be computed and accessed in $\O(1)$-time.
\end{enumerate}
\end{assumption}

\begin{theorem}[Theorem 10.13, \cite{cklpgs22}]
\label{thm:convexFlow}
Let $G$ be a graph, and $\Bd \in \R^V$ be a demand vector.
Given a collection of computationally efficient cost functions on edges $\{c_e\}_e$ and their barriers $\{\psi_e\}_e$ (\cref{ass:cvxCost}), there is an algorithm that runs in $m^{1+o(1)}$ time and outputs a flow $\Bf \in \R^E$ that routes $\Bd$ and for any fixed constant $C>0$, 
\begin{align*}
    c(\Bf) \le \min_{\BB^\top \Bf^* = \Bd} c(\Bf^*) + \exp(-\log^C m)
\end{align*}
where $c(\Bf) \defeq \sum_{e\in E} c_e(\Bf_e)$. 
\end{theorem}

\subsection{Iterative Refinement}

At a high level, the iterative refinement framework introduced by \cite{akps19} approximates the Bregman Divergence of the function $|x|^p$ with something simpler.
\begin{definition}[Bregman divergence]
\label{def:bregman}
Given a differentiable convex function $g(\cdot)$ and any two points $\Bx, \By$ in its domain, we define its \emph{Bregman divergence} as
\begin{align*}
D_{g}(\By, \Bx) \defeq g(\By) - g(\Bx) - \l \g g(\Bx), \bDelta \r = \int_0^1 \int_0^t \l \bDelta, \g^2 g(\Bx + u\bDelta) \bDelta \r \mathrm{d} u \mathrm{d} t
\end{align*}
where $\bDelta \defeq \By - \Bx.$
\end{definition}

For any $p > 1$, \cite{akps19} shows that the Bregman Divergence of $|x+\delta|^p$ can be locally approximated by an error term $\gamma_p(\delta; |x|)$, which behaves quadratically in $\delta$ when $|\delta| \le |x|$ and as $|\delta|^p$ otherwise.

\begin{definition}[\cite{bubeck2018homotopy}]
\label{def:gamma}
For $x \in \R$, $f > 0$, and $p > 1$, we define 
\begin{align*}
    \gamma_p(x, f) \defeq \begin{cases}
    \frac{p}{2} f^{p-2} x^2 & |x| \le f, \\
    |x|^p - \left(1 - \frac{p}{2}\right) f^q & |x| > f.
    \end{cases}
\end{align*}
\end{definition}

For $1 < q \le 2$, \cite{AKPS22} approximates the Bregman divergence of $|x|^q$ using the $\gamma_q$ function we just defined.
\begin{lemma}[Lemma 2.14, \cite{AKPS22}]
\label{lemma:qIterRefine}
For $q \in (1, 2]$ and $f, x \in \R$, it holds that 
\[
\frac{q-1}{q 2^q} \cdot \gamma_q(x, |f|) \le |f+x|^q - |f|^q -q |f|^{q-2} f x \le 2^q \cdot \gamma_q (x, |f|)
\]
\end{lemma}

For $p \ge 2$, the Bregman divergence of $|x|^p$ can be approximated by a $x^2$ and a $|x|^p$ term.
\begin{lemma}[Lemma 2.5, \cite{AKPS22}]
\label{lemma:pIterRefine}
For $p \ge 2$ and $f, x \in \R$, it holds that 
\[
\frac{p}{8}|f|^{p - 2} x^2 + \frac{1}{2^{p+1}} |x|^p\le |f + x|^p - |f|^p - p|f|^{p - 2}f x \le 2p^2 |f|^{p - 2} x^2 + p^p |x|^p
\]
\end{lemma}

Here we present some facts about $\gamma_q$ functions that are helpful.
\begin{lemma}[Lemma 3.3, \cite{akps19}]
\label{lemma:gammaScaling}
For $q \in (1, 2]$, $f, x \in \R$, and $t \ge 1$, it holds that 
\[
t^q \cdot \gamma_q(x, |f|) \le \gamma_q(tx, |f|) \le t^2 \cdot \gamma_q(x, |f|). 
\]
\end{lemma}
\section{Iterative Refinement Algorithm}\label{sec:iterative_refinement}
\newcommand{\uinner}{\langle{\mb{1}, \mb{u}}\rangle}
\newcommand{\LHS}{\textnormal{LHS}}

In this section, we present the $\ell_{q, p}$-norm $k$-commodity flow algorithm based on iterative refinement and prove \cref{thm:LqpFlow}.

\mainLqpFlow*

In the rest of the paper, we use $\CE(\BF)$ to denote the objective of Problem~\eqref{eq:LqpFlow}, that is, $\CE(\BF) = \|\BW \BF\|_{q, p}^{pq}.$

Our iterative refinement algorithm is based on an approximation to the Bregman divergence of the objective $\|\BW \BF\|_{q, p}^{pq}.$
Because the objective can be decomposed into a summation of $m$ seperate terms, i.e. $\sum_e \Bw_e^{pq} \|\BF_e\|_q^{pq}$, we approximate the Bregman divergence for each edge separately.
That is, we prove the following lemma that approximates $\|\BF_e + \Bx_e\|_q^{pq}$ for each edge $e.$
This is the key technical lemma of this paper and we defer the proof to \cref{sec:proof_lemma_lr}.
\begin{restatable}{lemma}{LqpIterRefine}[$\ell_{q, p}$-Norm Iterative Refinement]
\label{lem:diff_upper_lower}
Given $1 < q \le 2 \le p$, and any $\Bf, \Bx \in \R^k$, we have
\begin{align}
    \|\Bf + \Bx\|_q^{pq} - \|\Bf\|_q^{pq} - p q \|\Bf\|_q^{q(p-1)} \langle |\Bf|^{q-2} \cdot \Bf, \Bx \rangle
    &\ge \frac{p (q-1)}{16} \|\Bf\|_q^{q(p - 1)} \cdot \gamma_q(\Bx, |\Bf|) + \frac{q-1}{pq - 1} \frac{1}{2^{pq+2}} \|\Bx\|_{pq}^{pq}, \label{ineq:qpLB}\\
    \|\Bf + \Bx\|_q^{pq} - \|\Bf\|_q^{pq} - p q \|\Bf\|_q^{q(p-1)} \langle |\Bf|^{q-2} \cdot \Bf, \Bx \rangle
    &\le \frac{7}{k} \|\Bf\|_q^{q(p-1)} \cdot \gamma_q(6kp \Bx, |\Bf|) + \frac{3(6pk)^{pq}}{k} \|\Bx\|_{pq}^{pq} \label{ineq:qpUB} 
\end{align}
where we write $\gamma_q(\Bx, \By) \defeq \sum_j \gamma_q(\Bx_j, \By_j)$ for vectors $\Bx, \By \in \R^k.$
\end{restatable}

Using \Cref{lem:diff_upper_lower}, we can define the \emph{residual function} $\CR(\Bx, \Bf)$ that upper bounds the difference in the objective, i.e., $\CR(\Bx, \Bf) \ge \|\Bf + \Bx\|_{q}^{pq} - \|\Bf\|_{q}^{pq}.$

\begin{definition}[Residual Problem]
\label{def:residual}
Given two vectors $\Bf, \Bx \in \R^k$, we define $\CR(\Bx; \Bf)$ as
\begin{align*}
\CR(\Bx; \Bf) \defeq p q \|\Bf\|_q^{q(p-1)} \langle |\Bf|^{q-2} \cdot \Bf, \Bx \rangle + \frac{7}{k} \|\Bf\|_q^{q(p-1)} \cdot \gamma_q(6kp \Bx, |\Bf|) + \frac{3(6pk)^{pq}}{k} \|\Bx\|_{pq}^{pq}
\end{align*}
In the setting of Problem~\ref{eq:LqpFlow}, given a feasible flow $\BF \in \R^{E \times k}$, we define the \emph{residual problem} w.r.t. $\BF$ as follows
\begin{align}
\label{eq:residualProb}
\min_{\BX \in \R^{E \times k}}\ &\CR(\BX; \BF) \defeq \sum_e \Bw_e^{pq} \CR(\BX_e; \BF_e)\\
\st \ &\BB^\top \BX = \mb{0} \nonumber
\end{align}
\end{definition}



Note that in the residual problem, the objective is decomposable for each coordinate $(e, j).$
The decomposability allows us to use the almost-linear time convex flow solver (\cref{thm:convexFlow}) to solve the residual problem to high accuracy.
The following lemma summarizes the algorithm and the proof is deferred to \cref{sec:residualSolve}. 
\begin{restatable}{lemma}{resSolve}
\label{lem:residualSolve}
Given any feasible flow $\BF \in \R^{E \times k}$ to Problem~\eqref{eq:LqpFlow}, there is a randomized algorithm that runs in $km^{1+o(1)}$ and outputs, for any $C > 0$, a $k$-commodity circulation $\BX$ such that
\begin{align*}
    \CR(\BX; \BF) \le \min_{\BB^\top \BX^* = \mb{0}} \CR(\BX^*; \BF) + \exp\left(-\log^C m\right)
\end{align*}
\end{restatable}



Now, we are ready to prove \cref{thm:LqpFlow} with the following \cref{alg:high_multicommodity}.
The algorithm starts with some initial flow and runs in $T$ iterations.
Each iteration, the algorithm update the flow $\BF^{(t+1)} \gets \BF^{(t)} + \BX^{(t)}$ with a near-optimal solution to the residual problem.
In other words, given a current flow $\BF$, the algorithm update the flow by solving a simpler residual problem which is an upper bound to the change in objective value.

\begin{algorithm}
\caption{High-Accuracy $\ell_{q,p}$-norm Multicommodity Flow \label{alg:high_multicommodity}}
\SetKwProg{Proc}{procedure}{}{}
\Proc{$\textsc{LqpNormFlow}(G, \Bw, \BD, q, p)$}{
Initialize the flow $\BF^{(0)} \in \R^{E \times k}$ such that $\BB^\top \BF^{(0)} = \mb{D}$ via \cref{lem:initialFlow}. \\ 
$T \gets O(p \lambda \log(m/\eps))$, where $\lambda$ comes from \cref{lem:Efx_Ef_lambda_x} \\
\For{$t = 0$ \KwTo $T - 1$}
{
Solve the residual problem~\eqref{eq:residualProb} w.r.t. $\BF^{(t)}$ via \cref{lem:residualSolve} and obtain the solution $\BX^{(t)}.$ \\
$\BF^{(t+1)} \gets \BF^{(t)} + \BX^{(t)}$ \label{line:Ftp1}\\
}
\Return $\BF^{(T)}$\\
}
\end{algorithm}

To analyze \cref{alg:high_multicommodity} and prove \cref{thm:LqpFlow}, we first relate the value of $\CR(\BX; \BF)$ to the change in objective value when updating $\BF$ with $\BX.$
\begin{restatable}{lemma}{ResidualApprox}
\label{lem:Efx_Ef_lambda_x}
For any $\BF, \BX \in \R^{E \times k}$, we have
\begin{align}
\CE(\BF + \BX) - \CE(\BF) &\le \CR(\BX; \BF)\text{, and} \label{ineq:Efx_Ef}\\
\CE(\BF + \lambda\BX) - \CE(\BF) &\ge \lambda \CR(\BX; \BF)\text{ where }\lambda \defeq O\left(k p \left(\frac{4032}{q-1}\right)^{\frac{1}{q-1}}\right). \label{ineq:Ef_lambda_x_Ef}
\end{align}
\end{restatable}
\begin{proof}
By the definition of $\mc{R}(\Bx)$ and Lemma~\ref{lem:diff_upper_lower}, the Inequality~\eqref{ineq:Efx_Ef} holds trivially. We now focus on proving Inequality~\eqref{ineq:Ef_lambda_x_Ef}.
In particular, we show that for any vector $\Bf, \Bx \in \R^{k}$, the following
\begin{align}\label{ineq:singleEdgeLB}
    \norm{\Bf + \lambda \Bx}_q^{pq} - \norm{\Bf}_q^{pq} \ge \lambda \CR(\Bf; \Bx)
\end{align}
Inequality~\eqref{ineq:singleEdgeLB} implies Inequality~\eqref{ineq:Ef_lambda_x_Ef} by taking the summation over all edges.

\cref{lem:diff_upper_lower} gives 
\begin{align*}
    &~ \|\Bf + \lambda \Bx\|_q^{pq} - \|\Bf\|_q^{pq} \\
    \ge &~ p q \|\Bf\|_q^{(p-1)q} \langle |\Bf|^{q-2} \cdot \Bf, \lambda \Bx \rangle + \frac{p (q-1)}{16} \|\Bf\|_q^{(p-1) q} \cdot \gamma_q(\lambda \Bx, |\Bf|) + \frac{q-1}{pq - 1} \frac{1}{2^{pq + 2}} \|\lambda \Bx\|_{pq}^{pq} 
\end{align*}
In addition, we have 
\[
\lambda \CR(\Bx) = \lambda p q \|\Bf\|_q^{(p-1) q} \langle |\Bf|^{q-2} \cdot \Bf, \Bx \rangle + \frac{7\lambda }{k} \|\Bf\|_q^{(p-1) q} \cdot \gamma_q(6kp \Bx, |\Bf|) + \frac{3 \lambda (6pk)^{pq}}{k} \|\Bx\|_{pq}^{pq}. 
\]
In order to prove \eqref{ineq:singleEdgeLB}, it suffices to ensure that 
\begin{align*}
\frac{p (q-1)}{16} \|\Bf\|_q^{(p-1) q} \cdot \gamma_q(\lambda \Bx, |\Bf|)  &\ge \frac{7\lambda }{k} \|\Bf\|_q^{(p-1) q} \cdot \gamma_q(6kp \Bx, |\Bf|), \\
\frac{q-1}{pq - 1} \frac{1}{2^{pq + 2}} \|\lambda \Bx\|_{pq}^{pq} &\ge \frac{3 \lambda (6pk)^{pq}}{k} \|\Bx\|_{pq}^{pq}. 
\end{align*}
We can consider each entry separately and \eqref{ineq:singleEdgeLB} follows if for any $f, x \in \R$, we have
\begin{align}
    \frac{p (q-1)}{16} \cdot \gamma_q(\lambda x, |f|) &\ge \frac{7 \lambda}{k} \cdot \gamma_q(6kpx, |f|), \label{ineq:first_lambda}\\
    \frac{q - 1}{pq - 1} \frac{1}{2^{pq + 2}} |\lambda x|^{pq} &\ge \frac{3 \lambda (6pk)^{pq}}{k} |x|^{pq}. \label{ineq:second_lambda}
\end{align}
Inequality~\eqref{ineq:first_lambda} follows from \cref{lemma:gammaScaling}
\[
\frac{p (q-1)}{16} \cdot \gamma_q(\lambda x, |f|) \ge \frac{p (q-1)}{16} \left(\frac{\lambda}{6kp}\right)^q \cdot \gamma_q(6kpx, |f|) \ge \frac{7 \lambda}{k} \cdot \gamma_q(6kpx, |f|), 
\]
if we set
\begin{align}\label{ineq:lambda_bound_1}
    \lambda \ge k p \left(\frac{4032}{q-1}\right)^{\frac{1}{q-1}} \ge 6kp
\end{align}
For Inequality~\eqref{ineq:second_lambda} to hold, we need to set 
\begin{align}\label{ineq:lambda_bound_2}
    \lambda \ge 1728 k \left(\frac{pq - 1}{q - 1}\right)^{\frac{1}{pq - 1}} p^{\frac{p q}{pq - 1}}
\end{align}

Observe that
\begin{align*}
    \left(\frac{pq-1}{q-1}\right)^{\frac{1}{pq-1}} &= (pq-1)^{\frac{1}{pq-1}} \left(\frac{1}{q-1}\right)^{\frac{1}{pq - 1}} \le e \left(\frac{1}{q-1}\right)^{\frac{1}{q-1}}\text{, and} \\ 
    p^{\frac{pq}{pq-1}} &= p^{1 + \frac{1}{pq - 1}} \le p^{1 + \frac{1}{p-1}} \le 2 p
\end{align*}
Combining the observation along with \eqref{ineq:lambda_bound_1} and \eqref{ineq:lambda_bound_2}, Inequality~\eqref{ineq:singleEdgeLB} follows if we set
\[
\lambda = O\left(k p \left(\frac{4032}{q-1}\right)^{\frac{1}{q-1}}\right)
\]
This concludes the proof.
\end{proof}

Using \cref{lem:Efx_Ef_lambda_x}, we now show that each iteration decreases the objective exponentially.
This is summarized by the following lemma
\begin{lemma}[Convergence Rate]
\label{lem:algoConvergence}
Let $\BF^*$ be the optimal solution to Problem~\eqref{eq:LqpFlow}.
At any iteration $t$ of \cref{alg:high_multicommodity}, we have
\begin{align*}
\CE(\BF^{(t+1)}) - \CE(\BF^*) \le \left(1 - \frac{1}{\lambda}\right) \left(\CE(\BF^{(t)}) - \CE(\BF^*)\right) + \exp(-\log^{C} m)
\end{align*}
\end{lemma}
\begin{proof}
Recall that $\BF^{(t+1)} = \BF^{(t)} + \BX^{(t)}$ and $\BX^{(t)}$ is a high accuracy solution to the residual problem w.r.t. $\BF^{(t)}.$
\cref{lem:Efx_Ef_lambda_x} and the optimality of $\BX^{(t)}$ yields
\begin{equation}\label{ineq:thm_2}
\begin{aligned}
\CE(\BF^{(t+1)}) - \CE(\BF^{(t)})
&\le \CR(\BX^{(t)}) \\
&\le \CR\left(\lambda^{-1}\left(\BF^* - \BF^{(t)}\right)\right) + \exp(-\log^{C} m) \\
&\le \frac{1}{\lambda}\left(\CE(\BF^*) - \CE(\BF^{(t)})\right) + \exp(-\log^{C} m)
\end{aligned}
\end{equation}

We can conclude the lemma as follows:
\begin{align*}
\CE(\BF^{(t+1)}) - \CE(\BF^*)
&= \CE(\BF^{(t+1)}) - \CE(\BF^{(t)}) + \CE(\BF^{(t)}) - \CE(\BF^*) \\
&\le \frac{1}{\lambda}\left(\CE(\BF^*) - \CE(\BF^{(t)})\right) + \exp(-\log^{C} m) + \CE(\BF^{(t)}) - \CE(\BF^*) \\
&= \left(1 - \frac{1}{\lambda}\right) \left(\CE(\BF^{(t)}) - \CE(\BF^*)\right) + \exp(-\log^{C} m)
\end{align*}

\end{proof}

Then, we show how to find an initial solution efficiently.
\begin{lemma}[Initial Flow]
\label{lem:initialFlow}
Given an instance of Problem~\eqref{eq:LqpFlow}, there is an algorithm that runs in $\O(pmk)$-time and finds a feasible flow $\BF^{(0)} \in \R^{E \times k}$ such that $\CE(\BF^{(0)}) \le 4m^{p+1}\CE(\BF^*)$ where $\BF^*$ is the optimal solution.
\end{lemma}
\begin{proof}
We let flow $\BF^{(0)}_j$ be an $(1+1/pq)$-approximate maximum concurrent flow on $(G, \Bw)$ that routes the demand $\BD.$
That is, $\BF^{(0)}$ is an $(1+1/pq)$-approximate solution to the following problem:
\begin{align*}
    \min_{\BB^\top \BF = \BD} \max_e \Bw_e \sum_j |\BF_{ej}|
\end{align*}
$\BF^{(0)}$ can be computed in $\O(pmk)$-time using the algorithm from \cite{sherman17}.
Let $\BF^*$ be the optimal solution to Problem~\eqref{eq:LqpFlow}.
We can view $\BF^*$ as a collection of $k$ single commodity flows and the approximation guarantee of $\BF^{(0)}$ yields
\begin{align}\label{ineq:initialFlow}
\max_e \Bw_e \norm{\BF^{(0)}_e}_1 \le \left(1 + \frac{1}{pq}\right) \max_e \Bw_e \norm{\BF^*_e}_1
\end{align}

We now analyze the approximation ratio of $\BF^{(0)}$.
Recall the fact that for any vector $\Bx \in \R^k$, we have $\norm{\Bx}_q \le \norm{\Bx}_1 \le m^{1-1/q}\norm{\Bx}_q.$
Combining the observation with Inequality~\eqref{ineq:initialFlow}, we have
\begin{align*}
\CE(\BF^*) 
&\le \CE(\BF^{(0)}) = \sum_e \Bw_e^{pq} \norm{\BF^{(0)}_e}_q^{pq} \\
&\le \sum_e \Bw_e^{pq} \norm{\BF^{(0)}_e}_{1}^{pq} \le m \left(\max_e \Bw_e \norm{\BF^{(0)}_e}_{1}\right)^{pq} \\
&\le m \left(1 + \frac{1}{pq}\right)^{pq} \left(\max_e \Bw_e \norm{\BF^*_e}_1\right)^{pq} \\
&\le 4m \sum_e \Bw_e^{pq} \norm{\BF^*_e}_1^{pq} \le 4m^{1 + pq - p} \sum_e \Bw_e^{pq} \norm{\BF^*_e}_{q}^{pq} \le 4m^{p+1} \CE(\BF^*)
\end{align*}
\end{proof}

We use \cref{lem:algoConvergence} to analyze the correctness of the algorithm and prove \cref{thm:LqpFlow}.
\begin{proof}[Proof of \cref{thm:LqpFlow}]
After $T = O(p \lambda \log(m / \varepsilon))$ iterations, we use \cref{lem:algoConvergence} to bound the objective of the final flow $\BF^{(T)}$ output by \cref{alg:high_multicommodity} as follows:
\begin{align*}
\CE(\BF^{(T)}) - \CE(\BF^*)
&\le \left(1 - \frac{1}{\lambda}\right)^T \left(\CE(\BF^{(0)}) - \CE(\BF^*)\right) + \exp(-\log^{C} m) \cdot \sum_{t = 0}^{T-1} \left(1 - \frac{1}{\lambda}\right)^t \\
&\le \exp\left(-\frac{T}{\lambda}\right) \left(\CE(\BF^{(0)}) - \CE(\BF^*)\right) + \lambda \cdot \exp(-\log^{C} m) \\
&\le \varepsilon \CE(\BF^*) + \varepsilon
\end{align*}
where the last inequality comes from $\CE(\BF^{(0)}) - \CE(\BF^*) \le 4m^{p+1} \CE(\BF^*)$ and set the constant $C$ sufficiently small in \cref{lem:residualSolve}.
Rearrangement yields that $\CE(\BF^{(T)})$ is at most $(1 + \varepsilon)\CE(\BF^*) + \varepsilon.$

We now analyze the runtime.
Initiazliation of $\BF^{(0)}$ takes $\O(pmk)$-time by \cref{lem:initialFlow}.
Each iteration takes $km^{1+o(1)}$-time due to \cref{lem:residualSolve} and there are $T = \O(p\lambda \log(1/\varepsilon))$ iterations.
The runtime bound follows.
\end{proof}

\subsection{Proof of Lemma~\ref{lem:diff_upper_lower}}
\label{sec:proof_lemma_lr}
In this section, we prove \cref{lem:diff_upper_lower}.

\LqpIterRefine*

We will use the following facts in the proof.

\begin{lemma}\label{lem:gamma_q_x_q}
For any $q \in (1, 2), x, f \in \R$, $\gamma_q(x, |f|) \le |x|^q$. 
\end{lemma}
\begin{proof}
When $|x| \le |f|$, we have
\[
\gamma_q(x, |f|) = \frac{q}{2} |f|^{q-2} x^2 \le \frac{q}{2} |x|^q. 
\] 
When $|x| > |f|$, 
\[
\frac{q}{2} |x|^q \le \gamma_q(x, |f|) = |x|^q - \left(1 - \frac{q}{2}\right) |f|^q \le |x|^q. 
\]
\end{proof}

\begin{lemma}\label{lem:Jensen}
For any $p \ge 2$ and $\Bu \in \R^k$, we have $\|\Bu\|_p^{p} \le \|\Bu\|_1^{p} \le k^{p-1}\|\Bu\|_p^{p}.$
\end{lemma}
\begin{proof}
This directly follows from $\|\Bu\|_1 \le k^{1 - \frac{1}{p}} \|\Bu\|_p.$
\end{proof}

We first prove the upper bound part of \cref{lem:diff_upper_lower} (Inequality~\eqref{ineq:qpUB}).
The idea is to first apply \cref{lemma:qIterRefine} to bound the difference between $\|\Bf + \Bx\|_q^q$ and $\|\Bf\|_q^q$ and get
\begin{align*}
    \|\Bf + \Bx\|_q^q \le \|\Bf\|_q^q + \delta
\end{align*}
for some $\delta \in \R.$
Then, we apply \cref{lemma:pIterRefine} to the $(\|\Bf\|_q^q + \delta)^p$ part and get the upper bound for $\|\Bf + \Bx\|_q^{pq}.$


\begin{proof}[Proof of the upper bound in Lemma~\ref{lem:diff_upper_lower}(Inequality~\eqref{ineq:qpUB})] 
For two vectors $\Bf, \Bx \in \R^k$, we have $\|\Bf+\Bx\|_q^q = \sum_{j=1}^{k} |\Bf_j + \Bx_j|^q$, in which by Lemma~\ref{lemma:qIterRefine} we have  
\[
\frac{q-1}{q 2^q} \cdot \gamma_q(\Bx_j, |\Bf_j|) \le |\Bf_j + x_j|^q - |\Bf_j|^q - q |\Bf_j|^{q-2} \Bf_j \Bx_j \le 2^q \cdot \gamma_q(\Bx_j, |\Bf_j|), \forall j
\]
Let $C \defeq \|\Bf\|_q^q$.
Taking summation over $j$ of the above inequality yields  
\begin{align}\label{ineq:f_plus_x_q_q}
\frac{q-1}{q 2^q} \cdot \gamma_q(\Bx, |\Bf|) \le \|\Bf + \Bx\|_q^q - C - q \l |\Bf|^{q-2} \cdot \Bf, \Bx \r \le 2^q \cdot \gamma_q(\Bx, |\Bf|)
\end{align}
We can rearrange the upper bound part of \eqref{ineq:f_plus_x_q_q} and have
\begin{align*}
    \|\Bf + \Bx\|_q^q &\le C + \l\mb{1}, \Bu\r\text{ , where} \\
    \Bu_j &\defeq q |\Bf_j|^{q-2} \Bf_j \Bx_j + 2^q \cdot \gamma_q(\Bx_j, |\Bf_j|) \forall j
\end{align*}

\cref{lemma:pIterRefine} bounds $\|\Bf + \Bx\|_q^{pq}$ as follows:
\begin{align*}
\|\Bf + \Bx\|_q^{pq}
&\le (C + \l\mb{1}, \Bu\r)^p \\
&\le C^p + p C^{p-1} \l\mb{1}, \Bu\r + 2p^2 C^{p - 2} \l\mb{1}, \Bu\r^2 + p^p \left|\l\mb{1}, \Bu\r\right|^p \\
&\le C^p + p C^{p-1} \l\mb{1}, \Bu\r + 2kp^2 C^{p - 2} \l\mb{1}, \Bu^2\r + k^{p-1} p^p \l\mb{1}, |\Bu|^p\r
\end{align*}
where the last inequality follows from $\l\mb{1}, \Bu\r \le \|\Bu\|_1$ and \cref{lem:Jensen}.

To conclude the desired upper bound (Inequality~\eqref{ineq:qpUB}), it suffices to prove that 
\begin{align*}
 &~ p C^{p-1} \l\mb{1}, \Bu\r + 2kp^2 C^{p - 2} \l\mb{1}, \Bu^2\r + k^{p-1} p^p \l\mb{1}, |\Bu|^p\r \\
 \le &~ p q C^{p-1} \l |\Bf|^{q-2} \cdot \Bf, \Bx \r + \frac{7}{k} C^{p-1} \cdot \gamma_q(6kp \Bx, |\Bf|) + \frac{3 (6kp)^{pq}}{k} \|\Bx\|_{pq}^{pq}
\end{align*}
It even suffices to prove that the inequality holds for each coordinate $j \in [k]$, i.e., to prove that
\begin{equation}
\label{ineq:target_upper}
\begin{aligned}
&~ p C^{p-1} \Bu_j + 2 k p^2 C^{p-2} \Bu_j^2 + k^{p-1} p^p |\Bu_j|^p \\
\le &~ p q C^{p-1} |\Bf_j|^{q-2} \Bf_j \Bx_j + \frac{7}{k} C^{p-1} \cdot \gamma_q(6kp \Bx_j, |\Bf_j|) + \frac{3(6kp)^{pq}}{k} |\Bx_j|^{pq}
\end{aligned}    
\end{equation}
In the rest of the proof, we ignore the subscript $j$ and write $u, f$ and $x$ to refer to $\Bu_j, \Bf_j$ and $\Bx_j$ respectively.

We write the LHS of \eqref{ineq:target_upper} using the definition of $\Bu_j$ as
\begin{align*}
\LHS &\defeq p C^{p-1} u + 2k p^2 C^{p-2} u^2 + k^{p-1} p^p u^p\text{ , where} \\
u &= q |f|^{q-2} f x + 2^q \cdot \gamma_q(x, |f|)
\end{align*}

Next, we do a case analysis on the value of $\gamma_q(x, |f|)$.
\begin{enumerate}[(1)]
\item $|x| > |f|$

In this case, $\frac{q}{2} |x|^q \le \gamma_q(x, |f|) \le |x|^q$, $|f|^{q-1} |x| \le |x|^q$, and then 
\[
|u| \le q |f|^{q-1} |x| + 2^q \cdot \gamma_q(x, |f|) \le (2+2^q) \cdot \gamma_q(x, |f|) \le 6 \cdot \gamma_q(x, |f|) \le 6 |x|^q. 
\]
Plugging the upper bound on $|u|$ further bounds the LHS by:
\begin{equation}
\label{ineq:LHS_case1}
\begin{aligned}
    \LHS &\le p q C^{p-1} |f|^{q-2} f x + 4 p C^{p-1} \cdot \gamma_q(x, |f|) + 2 k p^2 C^{p-2} (6 \cdot \gamma_q(x, |f|))^2 + k^{p-1} p^p (6 |x|^q)^p \\
    &\le p q C^{p-1} |f|^{q-2} f x + \frac{2}{3 k} C^{p-1} \cdot \gamma_q(6kpx, |f|) + \frac{2}{k} (6kp)^2 C^{p-2} \cdot \gamma_q(x, |f|)^2 + \frac{(6kp)^p}{k} |x|^{pq}
\end{aligned}
\end{equation}
where the second inequality follows follows from \cref{lemma:gammaScaling}: $\gamma_q (6kpx, |f|) \ge (6kp)^q \cdot \gamma_q(x, |f|) \ge 6 k p \cdot \gamma_q(x, |f|)$.

Next, we show that the term $(6kp)^2 C^{p-2}\gamma_q(x, |f|)^2$ is consumed by either the $\gamma_q(6kpx, |f|)$ part or the $|x|^{pq}$ part.
If $6kp \cdot \gamma_q(x, |f|) \le C$, we have 
\begin{align}\label{ineq:6kp2_Cp2_1}
(6kp)^2 C^{p-2} \cdot \gamma_q(x, |f|)^2 \le 6kp C^{p-1} \cdot \gamma_q(x, |f|) \le C^{p-1} \cdot \gamma_q(6kpx, |f|)
\end{align}
otherwise, we have
\begin{align}\label{ineq:6kp2_Cp2_2}
(6kp)^2 C^{p-2} \cdot \gamma_q(x, |f|)^2 \le (6kp)^p \cdot \gamma_q(x, |f|)^p \le (6kp)^p |x|^{pq}
\end{align}
Combining \eqref{ineq:LHS_case1}, \eqref{ineq:6kp2_Cp2_1}, and \eqref{ineq:6kp2_Cp2_2} yields the following:  
    \begin{align}\label{ineq:LHS_upper_1}
       \LHS &\le p q C^{p-1} |f|^{q-2} f x + \frac{3}{k} C^{p-1} \cdot \gamma_q(6kpx, |f|) + \frac{3 (6kp)^p}{k} |x|^{pq}. 
    \end{align}
\item $|x| \le |f|$ 

In this case, $\gamma_q(x, |f|) = \frac{q}{2} |f|^{q-2} x^2 \le \frac{q}{2} |f|^{q-1} |x|$ and we have
\[
|u| \le q |f|^{q-1} |x| + 2^q \cdot \gamma_q(x, |f|) \le (q + q 2^{q-1}) |f|^{q-1} |x| \le 6 |f|^{q-1} |x|. 
\]
Plugging the upper bound on $u$ further bounds the LHS by:
\begin{equation}
\label{ineq:LHS_upper}
\begin{aligned}
    \LHS &\le p q C^{p-1} |f|^{q-2} f x + 4 p C^{p-1} \cdot \gamma_q(x, |f|) + 2 k p^2 C^{p-2} (6 |f|^{q-1} |x|)^2 + k^{p-1} p^p (6 |f|^{q-1} |x|)^p \\
    &\le p q C^{p-1} |f|^{q-2} f x + \frac{2}{3 k} C^{p-1} \cdot \gamma_q(6kpx, |f|) + \frac{2}{k} C^{p-2} (|f|^{q-1} 6 k p |x|)^2 + \frac{1}{k} (|f|^{q-1} 6 k p |x|)^p
\end{aligned}
\end{equation}
where the second inequality follows follows from $\gamma_q(6kpx, |f|) \ge 6 k p \cdot \gamma_q(x, |f|)$. 

To continue, we do another case analysis on the value of $\gamma_q(6kpx, |f|).$
\begin{enumerate}
    \item $|x| \le |6kpx| \le |f|$

    In this case, we have $\gamma_q(6kpx, |f|) = \frac{q}{2} |f|^{q-2}  (6kpx)^2$ and $|f|^{q-1} 6 k p |x| \le |f|^q \le C$, and we can bound the higher order parts of \eqref{ineq:LHS_upper} as follows:
    \begin{equation}
    \label{ineq:LHS_last_two}
    \begin{aligned}
    &~ \frac{2}{k} C^{p-2} (|f|^{q-1} 6 k p |x|)^2 + \frac{1}{k} (|f|^{q-1} 6 k p |x|)^p \\
    \le &~ \frac{2}{k} C^{p-2} (|f|^{q-1} 6 k p |x|)^2 + \frac{1}{k} C^{p-2} (|f|^{q-1} 6 k p |x|)^2 \\
    = &~ \frac{3}{k} C^{p-2} (|f|^{q-1} 6 k p |x|)^2
    \le \frac{3}{k} C^{p-1} |f|^{q-2} (6kpx)^2
    \le \frac{6}{k} C^{p-1} \cdot \gamma_q(6kpx, |f|), 
    \end{aligned}
    \end{equation}
    where the third step follows from $|f|^q \le C$ and the last step follows from $|f|^{q-2} (6kpx)^2 \le 2 \cdot \gamma_q(6kpx, |f|)$. 
    Combining \eqref{ineq:LHS_upper} with \eqref{ineq:LHS_last_two} gives 
    \begin{equation}\label{ineq:LHS_upper_2}
    \begin{aligned}
    \LHS &\le p q C^{p-1} |f|^{q-2} f x + \frac{2}{3 k} C^{p-1} \cdot \gamma_q(6kpx, |f|) + \frac{6}{k} C^{p-1} \cdot \gamma_q(6kpx, |f|) \\
    &\le p q C^{p-1} |f|^{q-2} f x + \frac{7}{k} C^{p-1} \cdot \gamma_q(6kpx, |f|). 
    \end{aligned}
    \end{equation}
    
    \item $|x| \le |f| < |6 k p x|$

    In this case, $\gamma_q(6kpx, |f|) \ge \frac{q}{2}|6kpx|^q$ and $|f|^{q-1} 6 k p |x| \le (6 k p |x|)^q \le 2 \cdot \gamma_q(6kpx, |f|)$.
    Using this, we bound the higher order parts of \eqref{ineq:LHS_upper} as follows:
    \begin{align}\label{ineq:LHS_case2b}
        \LHS \le p q C^{p-1} |f|^{q-2} f x + \frac{2}{3k} C^{p-1} \cdot \gamma_q(6kpx, |f|) + \frac{2}{k} C^{p-2} (6kpx)^{2q} + \frac{1}{k} |6kpx|^{pq}.
    \end{align}
    Next, we observe that $C^{p-2} (6kpx)^{2q}$ is consumed by either the $\gamma_q(6kpx, |f|)$ part or the $|6kpx|^{pq}$ part.
    If $(6kpx)^q \le C$, we have
    \begin{align}\label{ineq:Cp2_6kpx2q_1}
    C^{p-2} (6kpx)^{2q} \le C^{p-1} (6kpx)^q \le 2 C^{p-1} \cdot \gamma_q(6kpx, |f|); 
    \end{align}
    otherwise, 
    \begin{align}\label{ineq:Cp2_6kpx2q_2}
    C^{p-2} (6kpx)^{2q} \le (6kpx)^{pq}. 
    \end{align}
    Combining \eqref{ineq:LHS_case2b}, \eqref{ineq:Cp2_6kpx2q_1}, and \eqref{ineq:Cp2_6kpx2q_2}, we can obtain 
    \begin{align}\label{ineq:LHS_upper_3}
    \LHS \le p q C^{p-1} |f|^{q-2} f x + \frac{5}{k} C^{p-1} \cdot \gamma_q(6kpx, |f|) + \frac{3 (6kp)^{pq}}{k} |x|^{pq}. 
    \end{align}
\end{enumerate}
\end{enumerate}
\eqref{ineq:LHS_upper_1}, \eqref{ineq:LHS_upper_2}, and \eqref{ineq:LHS_upper_3} proves \eqref{ineq:target_upper} for each coordinate $j$ and concludes the upper bound.
\end{proof}

Next, we prove the lower bound part of \cref{lem:diff_upper_lower} (Inequality~\eqref{ineq:qpLB}).
Our main technical ingredient is the following lemma that shows the Hessian of the function $\|\Bx\|_q^{pq}$ can be approximated by a diagonal matrix.
Combining with the integral definition of the Bregman divergence (\cref{def:bregman}) and the lower bound for the Bregman divergence of $|x|^p$ for large $p$, we have the desired lower bound.


\begin{lemma}\label{lem:bregman}
For $1 < q \le 2 \le p$, the Hessian of the function $\ell(\Bx) \defeq \|\Bx\|_q^{pq}, \Bx \in \R^k$ is approximated by a diagonal matrix, that is, we have
    \[
    p q (q-1) \|\Bx\|_q^{(p-1)q} \cdot \diag(|\Bx|^{q-2}) \preceq \g^2 \ell(\Bx) \preceq p q (p q - 1) \|\Bx\|_q^{(p-1)q} \cdot \diag(|\Bx|^{q-2})
    \]
\end{lemma}

\begin{proof}
We first write out the explicit expression of the gradient and the Hessian of $\ell(\Bx):$
\begin{align*}
\g \ell(\Bx) &= p q \|\Bx\|_q^{(p-1)q} |\Bx|^{q-2} \cdot \Bx, \\
\g^2 \ell(\Bx) &= p q (q-1) \|\Bx\|_q^{(p-1)q} \cdot  \diag(|\Bx|^{q-2}) + p (p-1) q^2 \|\Bx\|_q^{(p-2)q} (|\Bx|^{q-2} \cdot \Bx) (|\Bx|^{q-2} \cdot \Bx)^\top
\end{align*}
This follows from, for any coordinates $i, j \in [k]$, 
\begin{align*}
    \frac{\partial \ell(\Bx)}{\partial \Bx_i}
    &= p q \|\Bx\|_q^{(p-1)q} |\Bx_i|^{q - 2} \Bx_i, \\
    \frac{\partial^2 \ell(\Bx)}{\partial \Bx_i^2} 
    &= p q (q-1) \|\Bx\|_q^{(p-1)q} |\Bx_i|^{q-2} + p (p-1) q^2 \|\Bx\|_q^{(p-2)q} |\Bx_i|^{2q-2}, \\
    \frac{\partial^2 \ell(\Bx)}{\partial \Bx_i \partial \Bx_j}
    &= p (p-1) q^2 \|\Bx\|_q^{(p-2)q} |\Bx_i|^{q-2} \Bx_i |\Bx_j|^{q-2} \Bx_j. 
\end{align*}
The lower bound follows directly from the expression of $\g^2 \ell(\Bx).$

For the upper bound, we have, for any vector $\Bv \in \R^k$, that
\begin{align*}
    \Bv^\top \g^2 \ell(\Bx) \Bv &= p q (q-1) \|\Bx\|_q^{(p-1)q} \cdot \Bv^\top \diag(|\Bx|^{q-2}) \Bv + p (p-1) q^2 \|\Bx\|_q^{(p-2) q} \l|\Bx|^{q-2} \cdot \Bx, \Bv\r^2 \\
    &\le p q (q-1) \| \Bx \|_q^{(p-1) q} \sum_{i} |\Bx_i|^{q-2} \Bv_i^2 + p (p-1) q^2 \|\Bx \|_q^{(p-2) q} \|\Bx\|_q^q \sum_{i} |\Bx_i|^{q-2} \Bv_i^2 \\
    &= p q (p q - 1) \|\Bx\|_q^{(p-1) q} \sum_{i} |\Bx_i|^{q-2} \Bv_i^2
\end{align*}
where the second step follows from Cauchy-Schwarz inequality 
\begin{align*}
\l|\Bx|^{q-2} \cdot \Bx, \Bv\r^2
\le \norm{|\Bx|^{\frac{q - 2}{2}} \Bx}_2^2\norm{|\Bx|^{\frac{q-2}{2}}\Bv}_2^2
= \norm{\Bx}_q^q  \sum_{i} |\Bx_i|^{q-2} \Bv_i^2
\end{align*}
\end{proof}

Now we prove the lower bound. 

\begin{proof}[Proof of the lower bound in Lemma~\ref{lem:diff_upper_lower} (Inequality~\eqref{ineq:qpLB})]
By \eqref{ineq:f_plus_x_q_q}, we have 
\begin{align*}
\|\Bf + \Bx\|_q^q &\ge C + q \langle |\Bf|^{q-2} \cdot \Bf, \Bx \rangle + \frac{q-1}{q 2^q} \cdot \gamma_q(\Bx, |\Bf|) \ge C + G + B\text{ , where} \\
C &\defeq \|\Bf\|_q^q, \\
G &\defeq q \langle |\Bf|^{q-2} \cdot \Bf, \Bx \rangle, \\
B &\defeq \frac{q-1}{8} \cdot \gamma_q(\Bx, |\Bf|). 
\end{align*}

However, since $G$ might be negative and hence $C + G + B$, we cannot lower bound $\|\Bf + \Bx\|_q^q$ by $|C + G + B|$ and apply \cref{lemma:pIterRefine} on $|C + G + B|^p.$
To bypass the issue, we consider two cases based on the value of $G.$
\begin{enumerate}[(1)]
    \item $C+G+B \ge 0$

    In this case, we can apply \cref{lemma:pIterRefine} to give a lower bound on $\|\Bf + \Bx\|_q^{pq}$ and this gives
    \begin{align*}
    \|\Bf + \Bx\|_q^{pq} &\ge (C + G + B)^p \\
    &\ge C^p + p C^{p-1} (G + B) + \frac{p}{8} C^{p-2} (G + B)^2 + \frac{1}{2^{p+1}} (G+B)^p \\
    &\ge C^p + p C^{p-1} (G + B) \\
    &= C^p + p C^{p-1} G + p C^{p-1} B
    \end{align*}
    \item $C+G+B < 0$

    In this case, we know that $G + B < -C$ and $p C^{p-1} (G + B) < - p C^p$, that is, it holds trivially that $\|\Bf + \Bx\|_q^{pq} \ge C^p + p C^{p-1} (G + B)$. 
\end{enumerate}
Hence, regardless of the value of $G$, we have
\begin{align}\label{ineq:lower_new_1}
\|\Bf + \Bx\|_q^{pq} - C^p -p C^{p-1} G \ge p C^{p-1} B = \frac{p(q-1)}{8} C^{p-1} \cdot \gamma_q(\Bx, |\Bf|). 
\end{align}

Finally, we will use the integral definition of Bregman divergence and \cref{lem:bregman} to establish the $\|\Bx\|_{pq}^{pq}$ part in the lower bound.
By defining $\ell(\Bf) = \|\Bf\|_q^{pq}$, we know that the LHS of the Inequality~\eqref{ineq:qpLB} is the Bregman divergence of $\ell(\cdot)$, i.e., 
\begin{align*}
\|\Bf + \Bx\|_q^{pq} - C^p -p C^{p-1} G = \ell(\Bf + \Bx) - \ell(\Bf) - \l\g \ell(\Bf), \Bx\r = D_{\ell}(\Bf + \Bx; \Bf)
\end{align*}
The integral-based definition of Bregman divergence (\cref{def:bregman}) yields
\begin{align*}
\|\Bf + \Bx\|_q^{pq} - C^p -p C^{p-1} G
&= D_{\ell}(\Bf + \Bx; \Bf) \\
&= \int_{0}^{1} \int_{0}^{t} \Bx^\top \g^2 \ell(\Bf + u \Bx) \Bx \mathrm{d} u \mathrm{d} t \\
&\ge p q (q-1) \int_{0}^{1} \int_{0}^{t} \|\Bf + u \Bx\|_q^{(p-1)q} \cdot \Bx^\top \diag(|\Bf + u \Bx|^{q-2}) \Bx \mathrm{d} u \mathrm{d} t \\
&= p q (q-1) \int_{0}^{1} \int_{0}^{t} \left(\sum_{j} |\Bf_j + u \Bx_j|^q\right)^{p-1} \left(\sum_{j} |\Bf_j + u \Bx_j|^{q-2} \Bx_j^2\right)\mathrm{d} u \mathrm{d} t
\end{align*}
where the inequality comes from the lower bound on $\g^2 \ell(\Bf)$ (\cref{lem:bregman}).

Fix any $u$, we have the following inequality based on Cauchy-Schwarz:
\begin{align*}
\left(\sum_{j} |\Bf_j + u \Bx_j|^q\right)^{p-1} \left(\sum_{j} |\Bf_j + u \Bx_j|^{q-2} \Bx_j^2\right) \ge \sum_{j} |\Bf_j + u \Bx_j|^{(p-1)q + q-2} \Bx_j^2 = \sum_{j} |\Bf_j + u \Bx_j|^{pq-2} \Bx_j^2
\end{align*}
Using this, we can further derive the lower bound as follows:
\begin{align*}
\|\Bf + \Bx\|_q^{pq} - C^p -p C^{p-1} G
&\ge p q (q-1) \int_{0}^{1} \int_{0}^{t} \left(\sum_{j} |\Bf_j + u \Bx_j|^q\right)^{p-1} \left(\sum_{j} |\Bf_j + u \Bx_j|^{q-2} \Bx_j^2\right)\mathrm{d} u \mathrm{d} t \\
&\ge p q (q-1) \int_{0}^{1} \int_{0}^{t} \sum_{j} |\Bf_j + u \Bx_j|^{pq-2} \Bx_j^2 \mathrm{d} u \mathrm{d} t \\
&= p q (q-1) \sum_j \int_{0}^{1} \int_{0}^{t} |\Bf_j + u \Bx_j|^{pq-2} \Bx_j^2 \mathrm{d} u \mathrm{d} t \\
&= p q (q-1) \sum_j \frac{1}{pq (pq - 1)} D_{|x|^{pq}}(\Bf_j + \Bx_j; \Bf_j),
\end{align*}
where the last step follows from that the integral is exactly the Bregman divergence w.r.t. the function $|x|^{pq}.$

Finally, we apply \cref{lemma:pIterRefine} to lower bound the Bregman divergence of $|x|^{pq}$ and have:
\begin{align*}
\|\Bf + \Bx\|_q^{pq} - C^p -p C^{p-1} G
&\ge p q (q-1) \sum_j \frac{1}{pq (pq - 1)} D_{|x|^{pq}}(\Bf_j + \Bx_j; \Bf_j) \\
&\ge \frac{(q - 1)}{pq (pq - 1)} \sum_j \frac{1}{2^{pq+1}} |\Bx_j|^{pq} = \frac{(q - 1)}{pq (pq - 1)} \frac{1}{2^{pq+1}} \|\Bx\|_{pq}^{pq}
\end{align*}

Combining this with \eqref{ineq:lower_new_1}, we concludes the lower bound as follows:
\begin{align*}
\|\Bf + \Bx\|_q^{pq} - \|\Bf\|_q^{pq} - pq \|\Bf\|_q^{q(p-1)} \l |\Bf|^{q - 2} \Bf, \Bx \r
&= \|\Bf + \Bx\|_q^{pq} - C^p -p C^{p-1} G \\
&\ge \frac{1}{2}\frac{p(q-1)}{8} C^{p-1} \cdot \gamma_q(\Bx, |\Bf|) + \frac{1}{2} \frac{(q - 1)}{pq (pq - 1)} \frac{1}{2^{pq+1}} \|\Bx\|_{pq}^{pq} \\
&= \frac{p(q-1)}{16} C^{p-1} \cdot \gamma_q(\Bx, |\Bf|) + \frac{(q - 1)}{pq (pq - 1)} \frac{1}{2^{pq+2}} \|\Bx\|_{pq}^{pq}. 
\end{align*}
\end{proof}

\section{Solving the Residual Problem}
\label{sec:residualSolve}
In this section, we prove \cref{lem:residualSolve}.

\resSolve*

At a high level, we will show that solving the residual problem is equivalent to solving $k$ instances of single commodity convex flow problems.
Each of them can be solved in $m^{1+o(1)}$-time via \cref{thm:convexFlow}.
In particular, we need to construct computationally efficient barriers for the edge cost functions.


\begin{proof}[Proof of \cref{lem:residualSolve}]
Recall the residual problem 
\begin{equation}\label{formulation:residual_new}
    \begin{aligned}
        \min_{\BX \in \R^{E \times k}} \ & \CR(\BX; \BF) \\
        \st \ & \BB^\top \BX = \mb{0}, 
    \end{aligned}
\end{equation}
From \cref{def:residual}, we know
\begin{align*}
\CR(\BX; \BF)
&= \sum_{e \in E} \sum_{j = 1}^k \Bw_e^{pq} \cdot \left(p q \|\BF_e\|_q^{q(p-1)} |\BF_{ej}|^{q-2} \BF_{ej} \BX_{ej} + \frac{7}{k} \|\BF_e\|_q^{q (p-1)} \gamma_q(6kp \BX_{ej}, |\BF_{ej}|) + \frac{3(6pk)^{pq}}{k} |\BX_{ej}|^{pq}\right)
\end{align*}
For each commodity $j \in [k]$, we write $\BX_j$ and $\BF_j$ to be the $j$-th column of $\BX$ and $\BF$ respectively, i.e., the flow that routes the commodity $j$.
The constraint $\BB^\top \BX = \mb{0}$ is equivalent to $\BB^\top \BX_j = \mb{0}$ for each $j.$
Thus, \eqref{formulation:residual_new} is equivalent to solve, for each commodity $j$, the following single commodity flow problem:
\begin{equation}\label{eq:resPerJ}
\begin{aligned}
\min_{\Bx \in \R^E} \sum_{e \in E} \ & \Bw_e^{pq} \cdot \left(p q \|\BF_e\|_q^{q(p-1)} |\BF_{ej}|^{q-2} \BF_{ej} \Bx_{e} + \frac{7}{k} \|\BF_e\|_q^{q (p-1)} \gamma_q(6kp \Bx_{e}, |\BF_{ej}|) + \frac{3(6pk)^{pq}}{k} |\Bx_{e}|^{pq}\right) \\
\st \ & \BB^\top \Bx = \mb{0}
\end{aligned}
\end{equation}

To use \cref{thm:convexFlow} to solve \eqref{eq:resPerJ}, we need to show that the objective is a sum of convex edge costs with efficient barriers.
For each edge $e$ and commodity $j$, we define the cost $c_{ej}(x)$ to be
\begin{align*}
c_{ej}(x) &=  A_{ej} x +  B_{ej}\gamma_q(6kp x, |\BF_{ej}|) +  C_{ej} |x|^{pq}\text{ , where} \\
A_{ej} &\defeq \Bw_e^{pq} \cdot p q \|\BF_e\|_q^{q(p-1)} |\BF_{ej}|^{q-2} \BF_{ej} \\
B_{ej} &\defeq \Bw_e^{pq} \cdot \frac{7}{k} \|\BF_e\|_q^{q (p-1)} \\
C_{ej} &\defeq \Bw_e^{pq} \cdot \frac{3(6pk)^{pq}}{k}
\end{align*}
Clearly $c_{ej}(x)$ is a convex function and the objective of \eqref{eq:resPerJ} is exactly $\sum_e c_{ej}(\Bx_e).$

We now present computationally efficient barriers for each $c_{ej}(x)$ and show that $c_{ej}(x)$ satisfies \cref{ass:cvxCost}.
For clarity, we ignore both subscripts $e$ and $j$ and use $f$ to denote $|\BF_{ej}|$.
That is, we will show that the following function satisfies \cref{ass:cvxCost} for any positive $A, B, C, f > 0$:
\begin{align*}
    c(x) = A \cdot x + B \cdot \gamma_q(6kp x; f) + C \cdot |x|^{pq}
\end{align*}
\cref{item:costQuasiPoly} holds because $p = \O(1).$

We now show \cref{item:costSC} using Theorem 9.1.1 from \cite{nemirovski04}.
It says that the function $\psi_c(x, y) = -\ln(y - c(x))$ is an $1$-self-concordance barrier for the epigraph $\CD_c \defeq \{(x, y) | c(x) \le y\}$ if there is some $\beta > 0$ such that $|c'''(x)| \le 3 \beta |c''(x) / x|$ holds for all $x \in \R.$
The derivatives of $c(x)$ have different forms depending on the value of $x$ due to the $\gamma_q$ function.
\begin{itemize}
\item If $|6kpx| \le f$, we know $\gamma_q(6kpx; f) = \frac{q}{2}f^{q-2}(6kp)^2x^2$ and
\begin{align*}
c'(x) &= A + B q f^{q-2} (6kp)^2 x + C pq |x|^{pq - 1} \sgn(x) \\
c''(x) &= B q f^{q-2} (6kp)^2 + C pq (pq-1) |x|^{pq - 2}\\
c'''(x) &= C pq(pq-1)(pq-2) |x|^{pq - 3}\sgn(x)
\end{align*}
In this case, we have
\begin{align*}
\left|\frac{c''(x)}{x}\right| = \frac{B q f^{q-2} (6kp)^2}{|x|} + C pq (pq-1) |x|^{pq - 3} \ge \frac{1}{pq-2}|c'''(x)|
\end{align*}
It suffices to set $\beta \ge (pq - 2) / 3.$
\item Otherwise, $|6kpx| > f$, we know $\gamma_q(6kpx; f) = |6kpx|^q - (1 - \frac{q}{2})f^q$, and
\begin{align*}
c'(x) &= A + B (6kp)^q q|x|^{q-1}\sgn(x) + C pq |x|^{pq - 1} \sgn(x) \\
c''(x) &= B (6kp)^q q(q-1)|x|^{q-2} + C pq (pq-1) |x|^{pq - 2}\\
c'''(x) &= B (6kp)^q q(q-1)(q-2)|x|^{q-3}\sgn(x) + C pq(pq-1)(pq-2) |x|^{pq - 3}\sgn(x)
\end{align*}
In this case, notice that $q - 2 < 0$ and we have
\begin{align*}
\left|\frac{c''(x)}{x}\right| &= B (6kp)^q q(q-1)|x|^{q-3} + C pq (pq-1) |x|^{pq - 3}\text{ , and}\\
|c'''(x)| &= \left|B (6kp)^q q(q-1)(q-2)|x|^{q-3} + C pq(pq-1)(pq-2) |x|^{pq - 3}\right| \\
&\le B (6kp)^q q(q-1)(2 - q)|x|^{q-3} + C pq(pq-1)(pq-2) |x|^{pq - 3}
\end{align*}
Setting $\beta \ge \max\{2 - q, pq-2\} / 3$ yields that
\begin{align*}
    3\beta\left|\frac{c''(x)}{x}\right| \ge  B (6kp)^q q(q-1)(2 - q)|x|^{q-3} + C pq(pq-1)(pq-2) |x|^{pq - 3} \ge |c'''(x)|
\end{align*}
\end{itemize}
We concludes \cref{item:costSC} with the barrier function $\psi_c(y, x) = -\ln(y - c(x)).$
\cref{item:costHessian} follows directly from \cref{item:costQuasiPoly} and \cref{item:costHessianCompute} follows using the explicit formula for $c(x), c'(x)$, and $c''(x).$

Thus, we can apply \cref{thm:convexFlow} to compute, for each $j$, a flow $\Bx^j$ in $m^{1+o(1)}$-time such that
\begin{align*}
    \sum_e c_{ej}(\Bx^{j}_e) \le \min_{\BB^{\top} \Bx^* = \mb{0}} \sum_e c_{ej}(\Bx^*_e) + \exp(-\log^C m)
\end{align*}
Define $\BX = [\Bx^1, \Bx^2, \ldots, \Bx^k].$
We have that $\BB^\top \BX = \mb{0}$ and $\BX$ is optimal to \eqref{formulation:residual_new} up to a $k\exp(-\log^C m)$-additive error, i.e.,
\begin{align*}
    \CR(\BX; \BF) \le \min_{\BB^\top \BX^*} \CR(\BX^*; \BF) + k \exp(-\log^C m)
\end{align*}
The total runtime is $k \cdot m^{1+o(1)}$ since we compute $\Bx^j$ for each $j \in[k].$

\end{proof}

\bibliography{Refs.bib}
\bibliographystyle{alpha}
\end{document}